\documentclass[12pt,a4paper]{amsart}
\usepackage{amsmath,amssymb,amsfonts,amsthm}
\usepackage{hyperref} 
\usepackage[mathscr]{eucal}
\usepackage{comment}

\date{6 (19) May 2015}

 \author{Sean~Hill}
\author{Ekaterina~Shemyakova}
\author{Theodore~Voronov}
\address{Department of Mathematics, SUNY  New Paltz, New Paltz, NY 12561-2443, USA}
\email{shemyake@newpaltz.edu}
\address{School of Mathematics, University of Manchester, Manchester, M60 1QD, UK\\
{\hphantom{hh}Dept. of Quantum Field Theory, Tomsk State University, Tomsk, 634050, Russia}}
\email{theodore.voronov@manchester.ac.uk}

\title[Darboux transformations on the superline]{Darboux transformations for differential operators on the superline}

\newtheorem{theorem}{Theorem}
\newtheorem*{lemma}{Lemma}

\theoremstyle{definition}

\renewcommand{\leq}{\leqslant}

\DeclareMathOperator{\Ker}{Ker}

 \DeclareMathOperator{\ord}{ord}

\DeclareMathOperator{\DO}{DO}

\newcommand{\p}{\partial}

\newcommand{\G}{{\Gamma}}
\newcommand{\f}{{\varphi}}

\renewcommand{\l}{{\lambda}}
\newcommand{\x}{{\xi}}

\newcommand{\const}{\mathrm{const}}

\newcommand{\Mf}{M_{\f}}

\newcommand{\xto}[1]{{\xrightarrow{#1}}}

\begin{document}

\maketitle

In this note, we show that an arbitrary Darboux transformation of a differential operator  on the superline factorizes into elementary Darboux transformations of order one. (All definitions are given in the text.) Similar statement holds for operators on the ordinary line.  Elementary super Darboux transformations    and their iterations  for particular operators were  considered before~\cite{liu-qp-darboux-for-skdv-lmp-1995}, \cite{li-nimmo:darboux-twist2010}.  Our main result   is   the full description of Darboux transformations for arbitrary operators on  the superline.

By the superline we mean a $1|1$-dimensional supermanifold. Let $x$ be an even coordinate and $\x$ be an odd coordinate. Denote $D=\p_{\x}+\x\p_x$, so  $D^2=\p$, $\p=\p_x$. The ring of differential  operators  on the superline is denoted $\DO(1|1)$.
An arbitrary operator in $\DO(1|1)$  can be uniquely written as $A=a_mD^m+a_{m-1}D^{m-1}+\ldots+a_0$, where the coefficients are functions of $x,\x$ and may also depend on some `external' even or odd variables.   We define   \emph{order}   by saying that for an operator  $A$ as  above $\ord A\leq m$. This   differs from the usual notion; e.g., $\ord D=1$, but  $\ord \p=2$. An operator $A$ of order $m$ is \emph{nondegenerate} if the top coefficient $a_m$ is invertible (in particular, even). Examples: $\p=D^2$ is nondegenerate, $\p_{\x}=D-\x D^2$ is not. Nondegenerate operators of even order are even, and of odd order, odd. They cannot be divisors of zero;  the set of all nondegenerate operators is  multiplicatively closed.
If $A$ is nondegenerate, then $A=a_m\cdot B$, where $B=D^m+b_{m-1}D^{m-1}+\ldots +b_0$ is monic. Division with remainder (from the left and from the right) is possible: for arbitrary $N$ and nondegenerate
$M$, there  exist unique    $Q_1, R_1$ and  $Q_2, R_2$ such that $N=MQ_1+R_1$ and $N=Q_2M+R_2$, where $\ord R_1, \ord R_2< \ord M$.

It is easy to see that for a nondegenerate operator $A$ of order $m$, $\dim\Ker A=n|n$ if $m=2n$ and $\dim\Ker A=n+1|n$ if $m=2n+1$. For example,   $\dim \Ker D=1|0$ (constants). Indeed, one can re-write the equation $D^m\f+a_{n-1}D^{m-1}\f+\ldots + a_0\f=0$ in the matrix form $D\psi=\G\psi$, with an odd matrix $\G$, where the components of $\psi$ are $\f,D\f, \ldots,D^{m-1}\f$. A solution of a linear differential equation of the form $D\psi=\G\psi$ is completely defined by an initial condition for the vector $\psi^0(x)$, where $\psi(x,\x)=\psi^0(x)+\x\psi^1(x)$, which gives the dimension of the solution space. (Note that for a degenerate operator $A$,    $\dim\Ker A$ can be infinite. Example: $A=\p_{\x}$.)

We see that nondegenerate operators on the superline  are similar in many aspects with ordinary differential operators.

For an   invertible function $\f$, we define  the   operator $\Mf:=D-D\ln \f=\f\circ D\circ \f^{-1}$.  Then $\Mf\f=0$ and $\Ker\Mf$ is spanned by $\f$. Every monic first-order operator has this form.  
Let  $L$ be an arbitrary nondegenerate operator of order $m$. Let $\f$ be   an  even   solution  of the  equation $L\f=0$.
In the sequel we shall   act  formally and assume that it is possible to divide by $\f$. A kind of ``B\'{e}zout's theorem'' holds:   $L$ is divisible by $\Mf$, so that $L=L'\Mf$, for a nondegenerate operator $L'$ of order $m-1$. (Indeed, $L=L'\Mf + \psi$ for some function $\psi$, hence $0=L\f=\psi\f$, so $\psi=0$.) By induction, there is a factorization of $L$ into first-order operators, $L=a\cdot M_{\f_1}\cdot\ldots\cdot M_{\f_m}$.

Consider nondegenerate operators of order $m$ with the some fixed principal symbol (i.e., the top coefficient $a_m$). For simplicity let $a_m=1$; the general case is similar. For two such operators $L_0$ and $L_1$,   a differential operator $M$ of an arbitrary order $r$ defines a \emph{Darboux transformation} $L_0\to L_1$ if the \emph{intertwining relation} $ML_0=L_1M$ holds. (There is a more general notion based on   intertwining relations of the form $NL_0=L_1M$ with possibly different $N$ and $M$, but we do not consider it here.) By definition, the order of $M$ is the  \emph{order}   of the Darboux transformation. The operator $M$ is automatically nondegenerate (in the considered case, monic). Note that $L_1$, if exists, is defined uniquely by $L_0$ and $M$. If no confusion is possible, we write $L_0\xto{M} L_1$. The problem is, for a given $L_0$, to find all $M$ defining its Darboux transformations. Darboux transformations can be composed and form a category: if $L_0\xto{M_{10}} L_1$ and $L_1\xto{M_{21}} L_2$, then $L_0\xto{M_{20}} L_2$ where $M_{20}:=M_{21}M_{10}$.


\begin{lemma}
 Every first-order Darboux transformation $L_0\to L_1$ is given by an operator $M=\Mf$, where $\f$ is an even eigenfunction of  $L_0$ with some 
 eigenvalue $\l$. 
\end{lemma}
\begin{proof}
 Let $\Mf L_0=L_1\Mf$ for some   $\Mf$. Divide $L_0$  by $\Mf$ from the right, so that $L_0=Q\Mf + f$ with some function $f$. Hence $\Mf\circ (Q\Mf + f)=L_1\Mf$. By applying both sides to $\f$, we obtain $\Mf(f\f)=0$; hence $f=\l=\const$. (Note that $\l$ may be even or odd depending on $m$.) Therefore $L_0\f=\l\f$. Conversely, if $L_0\f=\l\f$, we similarly deduce that $L_0=Q\Mf + \l$. Hence $\Mf L_0=L_1\Mf$   for $L_1:=\Mf Q + \l$\,.
\end{proof}

First-order Darboux transformations defined by the operators $M=\Mf$ as above are called \emph{elementary} Darboux transformations on the superline. (They are analogous to `Levy' or `Wro\'{n}ski' transformations on the line given by $M=\p-\p\ln\f$.) We have observed that all such Darboux transformations are given by `changing order in an incomplete factorization', $L_0=Q\Mf + \l \ \to\   L_1=\Mf Q + \l$\,.

\begin{theorem}\label{main}
 Every Darboux transformation $L_0\xto{M} L_r$ of order $r$ is the composition of $r$   elementary first-order transformations: $L_0\xto{M_{\f_1}} L_1\xto{M_{\f_2}} L_2\xto{M_{\f_3}} \ldots \xto{M_{\f_r}} L_r$.
\end{theorem}
\begin{proof}
 Suppose $ML_0=L_rM$ for an operator $M$ of order $r$. It can be decomposed into first-order factors, $M=M_1\cdot\ldots\cdot M_r$, where $M_i=M_{\f_i}$ for some functions $\f_i$. Note that this does not suffice \emph{per se}, because we would also need to find the intermediate operators $L_1, \ldots, L_{r-1}$ related by the corresponding  Darboux transformations. We proceed by induction. Suppose $r>1$. From the intertwining relation, $L_0(\Ker M)\subset \Ker M$. Recall $\dim \Ker M=s|s \ \text{or} \ s+1|s$, if $r=2s \ \text{or} \  r=2s+1$. Take an invertible eigenfunction $\f$ of $L_0$ in $\Ker M$. Then $M=M'\Mf$, where $M'$ is a nondegenerate operator of order $r-1$, and $\Mf$ defines a Darboux transformation $L_0\to L_1$. We shall show that $M'$ defines a Darboux transformation $L_1\to L_r$. Indeed, from $M'\Mf L_0=L_rM'\Mf$ and  $\Mf L_0=L_1\Mf$, we obtain $M'L_1\Mf=L_rM'\Mf$. This implies $M'L_1=L_rM'$, as $\Mf$ is a non-zero-divisor. Thus $L_0\xto{M}L_r$ is factorized into $L_0\xto{\Mf}L_1\xto{M'}L_r$ where $\ord M'=r-1$, and this completes the inductive step.
\end{proof}

An analog of Theorem~\ref{main} with a similar proof holds   for the classical case of  operators on the   line.  Elementary transformations in that case are the Levy   transformations. We could not find this statement in the literature, though it should be known to experts. Note that in the fundamental monograph~\cite{matveev-salle:book1991}, Darboux transformations are  defined as  iterations of Levy transformations. Approach based on intertwining relations (which can be traced back to Darboux~\cite{darboux:lecons1889-2})  was used in~\cite{veselov-shabat:dress1993}, \cite{bagrov-samsonov:factorization1995}, \cite{samsonov:factorization1999}, where factorization of Darboux transformations  for the Sturm--Liouville   operator   was proved.

For 2D operators,   the more general intertwining relation $NL_0=L_1M$ has to be used. The  theory in 2D is richer because of different types of Darboux transformations (such as Laplace transformations besides Wronskian transformations). 
Factorization of Darboux transformations for 2D Schr{\"o}dinger operator conjectured by Darboux was established in~\cite{shemy:complete,shemy:fact};  new  invertible transformations were found in~\cite{shemy:invert}. We hope to study   elsewhere the  attractive possibilities that may open in the super version.



\end{document}